\documentclass[a4paper,11pt]{article}
\usepackage{amsmath}
\usepackage{amsthm}
\usepackage{amsfonts}
\usepackage{amssymb}
\usepackage{geometry}

\newcommand{\overbar}[1]{\mkern1.5mu\overline{\mkern-1.5mu#1\mkern-1.5mu}\mkern 1.5mu}

\theoremstyle{plain}
\newtheorem{thm}{Theorem}[section]

\theoremstyle{definition}
\newtheorem{defn}{Definition}[section]

\title{On properties of Karlsson Hadamards and sets of Mutually Unbiased Bases in dimension six}
\author{%
  Andrew S Maxwell$^1$, %
  Stephen Brierley$^2$ 
 \\ \\
$^1${\it Dept. of Mathematics, University of Bristol, Bristol BS8 1TW, UK}\\
$^2${\it Heilbronn Institute for Mathematical Research, Dept. of Mathematics,} \\{\it University of Bristol, Bristol BS8 1TW, UK}\\
}

\begin{document}

\maketitle

\begin{abstract}
The complete classification of all $6 \times 6$ complex Hadamard matrices is an open problem. The 3-parameter Karlsson family encapsulates all Hadamards that have been parametrised explicitly. We prove that such matrices satisfy a non-trivial constraint conjectured to hold for (almost) all $6 \times 6$ Hadamard matrices. Our result imposes additional conditions in the linear programming approach to the mutually unbiased bases problem recently proposed by Matolcsi et al. Unfortunately running the linear programs we were unable to conclude that a complete set of mutually unbiased bases cannot be constructed from Karlsson Hadamards alone.



\end{abstract}

\section{Introduction}

In prime-power dimensions there are several ingenious methods to construct a complete set of mutually unbiased bases (MUBs) making use of finite fields, the Heisenberg-Weyl group, generalised angular momentum operators, and identities from number theory. However, even for the smallest composite dimension $d = 6$, the existence of such a set remains an open problem (See \cite{durt+10} for a review). The long-standing conjecture (also open in the study of Lie algebras \cite{boykin+05}) is that in non-prime-power dimensions complete sets do not exist. Distinguishing quantum systems based on their number theoretic properties would be an unusual feature not seen in Classical mechanics. The additional symmetry in dimensions such as $d = 3 \times 3$ would appear to enlarge the boundary of the set of quantum correlations (potentially) resulting in larger violations of Bell inequalities.  

Two recent papers have made progress in proving the non-existence of a complete set of MUBs in dimension 6. The first regards the classification of all complex Hadamard matrices. A complete set of $d+1$ MUBs is equivalent to a set of $d$ complex Hadamard matrices plus the identity matrix. Thus finding all possible Hadamard matrices in dimension 6 would be a significant step towards solving the MUBs problem. Recently, Karlsson has found an extremely nice 3-parameter family \cite{karlsson}. Whilst this does not fullly classify all $6 \times 6$ Hadamard matrices \cite{szollosi10}, it includes all explicitly parametrised families (except the isolated Spectral matrix \cite{tao04}). Moreover, Karlsson's parametrisation using Modius transformations is succinct and has allowed us to perform the calculations here.

The second result we use is due to Matolcsi et al, who demonstrate that sets of complex Hadamard matrices can be expressed as elements from a compact abelian group \cite{matolcsi10, matolcsi+12}. Fourier analysis then allows us to re-write the problem using the dual group where existence (or not) corresponds to a \emph{linear} program. This method was used to classify all complete sets of MUBs in dimensions $d=2 \ldots 5$ and obtain some partial results in $d=6$ \cite{matolcsi+12}. 

In this paper, we prove that Karlsson Hadamard matrices satisfy a previously unknown constraint motivated by the linear programming approach of Matolcsi et al. Feeding this new condition in to the linear program, we attempt to prove that a complete set of MUBs cannot be constructed from Karlsson Hadamards alone. Unfortunately we were unsuccessful in this endeavour but believe this is a promising avenue of research as our result can be easily combined with other conditions. Previous attempts to prove non-existence have required extensive computational power \cite{brierley+10, Jaming+12}. The linear programming approach is very appealing from this perspective. Using sparse matrix methods, we had no problems running the program in MATLAB on a desktop PC.

\section{Conditions on unbiased sets of complex Hadamard matrices}
A $d \times d$ unitary matrix, $H$, is called a \emph{complex Hadamard matrix} if its entries are all complex phases, 
$|H_{ij}|^2=1 $, for $i,j=1 \ldots d$. Any set of $r+1$ MUBs has an equivalent representation in terms of Hadamard matrices $\{ I, H_1, \ldots , H_r \}$, where $I$ is the $d \times d$ identity element (see \cite{Bengtsson+13} for a summary of equivalences on sets of MUBs).

There is a zoo of Hadamard matrices in dimension six, easily accessible online \cite{onlinelist}. At the time of writing, there where 10 families with various properties such as being symmetric or continuously deformed from the Fourier matrix. There is a 4-parameter family of Hadamards \cite{szollosi10} and furthermore, it is known that 4 parameters are sufficient for a complete classification \cite{zyczkowski06}. Unfortunately, the construction given in \cite{szollosi10} is implicit making it difficult to work with.

Of particular importance to the classification of $6 \times 6$ Hadamard matrices is Karlsson's 3-parameter family that includes all explicitly parametrised Hadamards with the exception of the spectral matrix \cite{karlsson}. The matrices are called $H_2$ reducible meaning that there are made from nine $2\times2$ Hadamard matrices. Imposing the $H_2$ reducible condition makes it possible to derive the following simple form of the Hadamard matrices \cite{karlsson}.

\begin{equation}
\label{eq:Karlsson_Standard}
K= \begin{pmatrix}
		F_2 & Z_1 & Z_2\\
		Z_3 & a & b\\
		Z_4 & c & d
		\end{pmatrix}\, ,
\end{equation}
where $F_2,Z_{1,2,3,4},a,b,c,d$ are $2\times2$ matrices defined by 
\begin{align*}
F_2&=\begin{pmatrix}
					1 & 1\\
					1 & -1
					\end{pmatrix}&
Z_{1,2}&=\begin{pmatrix} 1&1\\z_i&-z_i\\ \end{pmatrix}&
Z_{3,4}& = \begin{pmatrix} 1&z_i\\1&-z_i\\ \end{pmatrix}\\
a &= \frac{1}{2}Z_3AZ_1 & b&=\frac{1}{2}Z_3BZ_2 &A&=\begin{pmatrix}
									A_{11}&A_{12}\\
									\overbar{A_{12}}&-\overbar{A_{11}}
									\end{pmatrix}\\
c &= \frac{1}{2}Z_4BZ_1 & d&=\frac{1}{2}Z_4AZ_2 &B&=-F_2-A
\end{align*}
The matrix elements $A_{11}$ and $A_{12}$ are given by 
\begin{align*}
A_{11}&=-\frac{1}{2}+i\frac{\sqrt{3}}{2}\left(\cos(\theta)+e^{-i\varphi}\sin(\theta)\right) &A_{12}&=-\frac{1}{2}+i\frac{\sqrt{3}}{2}\left(-\cos(\theta)+e^{i\varphi}\sin(\theta)\right) \, .
\end{align*}
leaving a total of six undefined parameters $z_{1,2,3,4},$ $\theta$ and $\phi$. However, the $z$ variables are related via the mobi\"{u}s transformations 
\begin{align}
\label{eq:mobius1}
z_{3}^{2} &= \mathcal{M}_{A}\left( z_1^2\right) &z_3^2&=\mathcal{M}_B\left( z_{2}^{2}\right)\\
\label{eq:mobius2}
z_{4}^{2} &= \mathcal{M}_{A}\left( z_2^2\right) &z_4^2&=\mathcal{M}_B\left( z_{1}^{2}\right)\\
\intertext{where}
\label{eq:mobius3}
\mathcal{M}(z) &= \frac{\alpha z-\beta}{\overbar{\beta}z-\overbar{\alpha}} & \mathcal{M}^{-1}(w) &=\frac{\overbar{\alpha} w-\beta}{\overbar{\beta}w-\alpha} 
\end{align}
and $\alpha_A = A_{12}^2$, $\beta_A = A_{11}^2$, $\alpha_B = B_{12}^2$ and $\beta_B = B_{11}^2$. Leaving $z_1$, $\theta$ and $\phi$ as the three free parameters.

\section{Fourier Analytic approach}
In two recent papers, \cite{matolcsi10} and \cite{matolcsi+12}, Matolcsi et al. explain how to use techniques from Fourier analysis to attack the mutually unbiased bases problem. We briefly summarise their ideas for completeness and to define the notation used later. 

We can think of a set of $r$ Hadamards $\{ H_1, \ldots, H_r  \}$ as $D=rd^2$ phases, one from each matrix entry. Therefore, the set $\{ H_1, \ldots, H_r  \}$ corresponds to a point on the $D$-dimensional torus $ \mathbb{T}^D$. We can multiply two points in $ \mathbb{T}^D$ by multiplying component-wise so with this binary operation, $ \mathbb{T}^D $ forms a group. Since $\mathbb{T}^D$ is a locally compact abelian group, it has a dual $\hat{\mathbb{T}}^D = \mathbb{Z}^D$, characters of which, $\gamma = (r_1, \ldots r_D)\in \mathbb{Z}^D$, act on the group elements $v=(v_1 \ldots v_D)\in \mathbb{T}^D$ by exponentiation
$$
\gamma(v)=v_1^{r_1} \ldots v_D^{r_D} \, .
$$
We can now define the Fourier transform of a set $ S\subset \mathbb{T}^D  $ in the usual way  
$$
\hat{S}(\gamma) = \sum_{s \in S}s^{\gamma} \, .
$$ 

A natural set to consider is given by the columns $\mathbf{c}_1 ,\ldots \mathbf{c}_{d}$ of a Hadamard matrix since we can later impose conditions such as orthogonality (and unbiasedness). Its Fourier transform is given by the function
$$
\sum^{d}_{k=1} \mathbf{c}_{k}^{\gamma} = \sum^{d}_{k=1} \prod_{i=1}^d  (H_j)_{ik}^{\gamma} \, .
$$
In what follows, it will be more convenient to define the permutation invariant version of this function
\begin{equation} \label{Def:smallg}
g_j(\gamma):= \frac{1}{d!}\sum_{\sigma}\sum^{d}_{k=1} \mathbf{c}_{k}^{\sigma(\gamma)}=\frac{1}{d!}\sum_{\sigma}\sum^{d}_{k=1}\prod^{d}_{i=1}(H_j)_{ik}^{\sigma(\gamma)} \, ,
\end{equation}
where we sum over all $d!$ permutations $\sigma$. Since $\overbar{g_j}=\sum_{k=1}^{d}\mathbf{c}_k^{-\gamma}$, the modulus of $g_j(\gamma)$ can be written as
\begin{equation}
G_j(\gamma):=|g_j(\gamma)|^2=\frac{1}{d!}\sum_{\sigma}\sum_{k,l=1}^d\left( \frac{\mathbf{c}_k}{\mathbf{c}_l} \right)^{\sigma(\gamma)} \, . \label{Def:BigGj}
\end{equation}

Now suppose we have a complete set of MUBs corresponding to $d$ complex Hadamard matrices. We can define the functions
\begin{align}
G(\gamma)&:=\sum^{d}_{j=1}|g_j(\gamma)|^2=\frac{1}{d!}\sum_{\forall \sigma}\sum^{d}_{j,k,l=1} \left(\frac{ \mathbf{c}_{k}}{\mathbf{c}_{l}}\right)^{\sigma(\gamma)}  \label{Def:bigG}
\end{align}
and,
\begin{align}
f(\gamma)&:=\sum^{d}_{j=1}g_{j}(\gamma) \label{Def:smallf}\\
F(\gamma)&:=|f(\gamma)|^2 \, . \label{Def:bigF}
\end{align}

\section{Constraints on sets of mutually unbiased Hadamard matrices}
The functions defined in Eqs. (\ref{Def:smallg}) to (\ref{Def:bigF}) satisfy certain conditions for the inputs to correspond to a set of mutually unbiased Hadamard matrices. Following \cite{matolcsi+12}, we derive constraints that hold for \emph{any} set of complex Hadamard matrices then in Theorem 1 present an additional non-trivial condition that we are able to prove holds for matrices from the Karlsson family. Taken together, these equalities and inequalities will later constrain a linear program. 

Let $\pi_r=(0,0,\dots,1,\dots,0)$ with the $r$th coordinate equal to 1.
\begin{align}
\sum_{r=1}^{d}G_{j}(\gamma+\pi_{r})&=\frac{1}{d!}\sum_{r=1}^{d}\left(\sum_{\sigma} \sum_{k,l=1}^{d}\left(\frac{c_{k}}{c_l}\right)^{\sigma(\gamma+\pi_{r})}\right) \nonumber\\
&=\frac{1}{d!}\sum_{r=1}^{d}\left(\sum_{\sigma} \sum_{k,l=1}^{d}\left(\frac{c_{k}}{c_l}\right)^{\sigma(\gamma)}\left(\frac{c_{k}}{c_l}\right)^{\sigma(\pi_{r})}\right) \nonumber\\
&=\frac{1}{d!}\sum_{\sigma} \sum_{k,l=1}^{d}\left(\frac{c_{k}}{c_l}\right)^{\sigma(\gamma)}\left(\sum_{r=1}^{d}\left(\frac{c_{k}}{c_l}\right)^{\sigma(\pi_{r})}\right) \label{cond:1}\\
&=\frac{1}{d!}\sum_{\sigma} \sum_{k,l=1}^{d}\left(\frac{c_{k}}{c_l}\right)^{\sigma(\gamma)}\sum_{r=1}^{d}\delta_{l,k} \nonumber \\
&=\frac{1}{d!} \sum_{\sigma} \sum_{l=k}^{d}d \nonumber \\
&=d^2 \, .\nonumber
\end{align}
Using (9) from \cite{matolcsi+12} and the above condition we can derive another equality constraint.
\begin{align}
\sum_{r \neq t}F(\gamma +\pi_r -\pi_t) &- \sum_{r \neq t} G(\gamma+\pi_r+\pi_t)=\frac{1}{d!}\sum_{\sigma}\sum_{r \neq t}\sum_{\mathbf{u},\mathbf{v}}\left(\frac{\mathbf{u}}{\mathbf{v}} \right)^{\sigma(\gamma+\pi_r -\pi_t)} \nonumber \\
&=\frac{1}{d!}\sum_{\sigma}\sum_{\mathbf{u},\mathbf{v}}\left(\frac{\mathbf{u}}{\mathbf{v}} \right)^{\sigma(\gamma)}\sum_{r \neq t}\left(\frac{\mathbf{u}}{\mathbf{v}} \right)^{\sigma(\pi_r-\pi_t)}=0 \, ,\label{eq:GFsum} 
\end{align}
where the sum over $\mathbf{u}$ and $\mathbf{v}$ is over all pairs of columns from different matrices. Together with $dG(\gamma)+\sum_{r \neq t}G(\gamma +\pi_r-\pi_t)=d^4$, equation (\ref{eq:GFsum}) now implies a further constraint
$$
dG(\gamma)+\sum_{r \neq t}F(\gamma +\pi_r -\pi_t)=d^4 \, .\label{cond:2}\\
$$

There are also some simple constraints coming from the function definitions
\begin{align}
&F(0)=d^4, & G(0)=d^3 \label{cond:3}\\
&0\le F(\gamma)\le d^4,   & 0 \le G(\gamma) \le d^3 \, ,
\end{align}

 and the Cauchy-Schwartz inequality implies that
$$
F(\gamma) \le d G(\gamma) \, . \label{cond:4}
$$

Matolcsi et al used the conditions given in Eqs. (\ref{cond:1}) - (\ref{cond:4}) to prove various structural results about mutually unbiased bases. They show that in dimensions $d = 2 \ldots 5$ a complete set must be composed of roots of unity only - from which it easily follows that the known complete sets are unique. Although these results were known before \cite{Bengtsson+13}, this new technique significantly reduces the difficulty in deriving the result in dimension $d=5$. The authors went on to conjecture that in dimension $d=6$, all complex Hadamard matrices (except the Spectral Hadamard) must satisfy an additional constraint motivated by the linear programming approach. Armed with this extra condition we could then be able to (finally) prove non-existence of a complete set in dimension six. We now present the main result of this paper proving Conjecture 2.3 of Matolcsi et al \cite{matolcsi+12} for the case of Karlsson Hadamard matrices. 
\begin{thm}
\label{Thm:ExtraConstraints}
Let $H_j$ be an element of the Karlsson Family of $6 \times 6$ complex Hadamard matrices defined in Section 2. For all permutations of $\boldmath{\rho}=(1,1,1,-1,-1,-1)$, we have,
\[g_{j}(\rho)=0 \, .\] 
\end{thm}

\begin{proof}[\bf{Proof of Theorem \ref{Thm:ExtraConstraints}}]
Due to the structure of the Karlsson Hadamards, many permutations of $\rho$ are equivalent. Swapping elements in $\rho$ corresponds to permuting rows of the Hadamard. Since a Karlsson Hadamard consists of 9 $2 \times 2$ sub-matrices, we can generate vectors equivalent to $\rho$ by swapping columns $(1,2)$, $(3,4)$ and $(5,6)$. In addition, we can also swap pairs of columns $(1,2)$ with $(3,4)$ etc. All vectors $\rho$ will have to have at least one alternating pair $(1,-1)$ as there are an odd number 1s and -1s. The remaining pairs must then either be alternating or  $(1,1)$ and $(-1,-1)$ This means there are only two non-equivalent permutations of $\rho = (1,1,1,-1,-1,-1)$ and $\rho = (1,-1,1,-1,1,-1)$. We treat the two cases separately.

The function $g_j(\rho)$ is given by  
\begin{align}
\label{eq:ProofExtraConstraints_Step1}
g(\rho)  & = \sum^{6}_{j=1} \prod^{6}_{i=1}(H)^{\rho_{i}}_{i j} \\
         & = 1+\sum^{6}_{j=2}( 1+\prod^{6}_{i=2}(H)^{\rho_{i}}_{i j})
\end{align}
where we drop the subscript $j$ since it simply refers to the label of the Hadamard matrix in the set. Using the definition for Karlsson matrices from equation (\ref{eq:Karlsson_Standard}), we have

\begin{align*}
g(\rho) = & 1 + (-1)^{\rho_2+\rho_4+\rho_6}z_{3}^{\rho_3+\rho_4} z_{4}^{\rho_5+\rho_6}&+ z_{1}(a_{11}^{\rho_3}a_{21}^{\rho_4}c_{11}^{\rho_5}c_{21}^{\rho_6}-a_{12}^{\rho_3}a_{22}^{\rho_4}c_{12}^{\rho_5}c_{22}^{\rho_6})\\
& &+z_{2}(b_{11}^{\rho_3}b_{21}^{\rho_4}d_{11}^{\rho_5}d_{21}^{\rho_6}-b_{12}^{\rho_3}b_{22}^{\rho_4}d_{12}^{\rho_5}d_{22}^{\rho_6})
\end{align*}

\subsection*{Case 1:  $\rho = (1,-1,1,-1,1,-1)$}
We use the fact that the sub-matrices $a,b,c,d$ are unitary to cancel pairs of terms. Substituting  $\rho = (1,-1,1,-1,1,-1)$, we have
\begin{align*}
g(\rho) = 1 -z_{3}^0z_{4}^0 
\,+\, & \overbar{z_{1}}(a_{11}\overbar{a_{21}}c_{11}\overbar{c_{21}}-a_{12}\overbar{a_{22}}c_{12}\overbar{c_{22}})\\
+\, & \overbar{z_{2}}((b_{11}\overbar{b_{21}}d_{11}\overbar{d_{21}}-b_{12}\overbar{b_{22}}d_{12}\overbar{d_{22}})
\end{align*}
\noindent
Since $a$ and $b$ are unitary matrices, $a_{11}\overline{a_{21}}+a_{12}\overline{a_{22}}=0$, and $b_{11}\overline{b_{21}}+b_{12}\overline{b_{22}}=0$,
\begin{equation*}
g_1(\rho) = \overbar{z_{1}}a_{11}\overbar{a_{21}}(c_{11}\overline{c_{21}}+c_{12}\overbar{c_{22}}) + \overline{z_{2}}b_{11}\overbar{b_{21}}(d_{11}\overbar{d_{21}}+d_{12}\overbar{d_{22}}) \, .
\end{equation*}
Since $c$ and $d$ are also unitary, it is easy to see that $g(\rho)=0$.

\subsection*{Case 2: $\rho = (1,1,1,-1,-1,-1)$}
The proof for $\rho = (1,1,1,-1,-1,-1)$ is more involved requiring us to make use of the parametrisation of Karlsson's family in terms of Modius transformations. Again using the fact that $a$ and $b$ are unitary matrices, we have that 

\begin{alignat*}{4}
g(\rho) &= 1 -\overbar{z_4}^{2} &+& z_{1}a_{11}\overbar{a_{21}}(\overbar{c_{11}}\overbar{c_{21}}+\overbar{c_{12}}\overbar{c_{22}})\\
&& +& z_{2}b_{11}\overbar{b_{21}}(\overbar{d_{11}}\overbar{d_{21}}+\overbar{d_{12}}\overbar{d_{22}})\\
\end{alignat*}
Now substituting the definition for $c$ and $d$ and simplifying using the Modius transformation given in Eqs. (\ref{eq:mobius1}), (\ref{eq:mobius2}) and (\ref{eq:mobius3}) you get,
\begin{align*}
\overbar{c_{11}}\overbar{c_{21}}+\overbar{c_{12}}\overbar{c_{22}}&=\overbar{z_1^2(B_{12}^2-\overbar{B_{11}^{2}}z_{4}^{2})}\\ \overbar{d_{11}}\overbar{d_{21}}+\overbar{d_{12}}\overbar{d_{22}}&=\overbar{z_2^2(A_{12}^2-\overbar{A_{11}^{2}}z_{4}^{2})} \, .\\
\intertext{Similarly, from the definition of $a$ and $b$, we find that}
a_{11}\overbar{a_{21}}&=\frac{1}{2}\overbar{z_{1}z_{3}}(A_{11}^2-A_{12}^2z_1^2+A_{11}A_{12}z_3+\overbar{A_{11}}A_{12})\\
b_{11}\overbar{b_{21}}&=\frac{1}{2}\overbar{z_{2}z_{3}}(B_{11}^2-B_{12}^2z_2^2+B_{11}B_{12}z_3+\overbar{B_{11}}B_{12})
\end{align*}
Substituting these last four expressions into $g$ gives
\begin{align*}
g(\rho) &= 1 -\overbar{z_4}^{2} &+&\frac{1}{2}\overbar{z_3}(A_{11}^2-A_{12}^2z_1^2+A_{11}A_{12}z_3+\overbar{A_{11}}A_{12})(\overbar{B_{12}}^2-B_{11}^{2}\overbar{z_{4}^{2}})\\
&&+&\frac{1}{2}\overbar{z_3}(B_{11}^2-B_{12}^2z_2^2+B_{11}B_{12}z_3+\overbar{B_{11}}B_{12})(\overbar{A_{12}^2}-A_{11}^{2}\overbar{z_{4}^{2}})\\
&=\frac{1}{2z_3z_{4}^2} \left[\right.  (2z_3z_4^2-2z_3) &+&(A_{11}^2-A_{12}^2z_1^2+A_{11}A_{12}z_3+\overbar{A_{11}}A_{12})(\overbar{B_{12}}^2z_4^2-B_{11}^{2}) \\
&&+&(B_{11}^2-B_{12}^2z_2^2+B_{11}B_{12}z_3+\overbar{B_{11}}B_{12})(\overbar{A_{12}^2}z_4^2-A_{11}^{2}\left.)\right] 
\end{align*}
We can drop the factor $\frac{1}{2z_3z_3^2}$ as it does not effect whether $g$ is zero or not, expanding terms
\begin{align*}
g(\rho) =  ( A_{11}^2 - & A_{12}z_1^2)(\overbar{B_{11}^2}z_4^2-B_{12}^2) + (B_{11}-B_{12}^{2}z_{2}^{2})(\overbar{A_{11}}z_{4}^{2}-A_{12}^{2})\\
  + z_3 [ & 2z_4^2-2 + (A_{11}\overbar{A_{12}} + \overbar{A_{11}}A_{12}z_1^2) (\overbar{B_{11}^2}z_4^2-B_{12}^2)\\
 & + (B_{11}\overbar{B_{12}} + \overbar{B_{11}}B_{12}z_1^2) (\overbar{A_{11}^2}z_4^2-A_{12}^2) ] \\
 \equiv & t_0 + z_3 t_1
\end{align*}

The coefficients of the zeroth and first power of $z_3$ can be examined separately and we now show that they are both zero $t_0=t_1=0$. The zeroth power coefficient, $t_0$, is 
$$
t_0 = (A_{11}^2-A_{12}z_1^2)(\overbar{B_{11}^2}z_4^2-B_{12}^2)+(B_{11}-B_{12}^{2}z_{2}^{2})(\overbar{A_{11}}z_{4}^{2}-A_{12}^{2})\\
$$
Expanding terms,
\begin{align*}
t_0=&z_4^{2}(A_{11}^2\overbar{B_{11}^2}-A_{12}^2\overbar{B_{12}^{2}}+B_{11}^2\overbar{A_{11}^2}-B_{12}^{2}\overbar{A_{12}^2})+A_{11}^2B_{11}^{2}-A_{11}^2B_{11}^{2}-A_{11}^{2}B_{12}^{2}+A_{11}^{2}B_{12}^{2} \\
=&z_4^{2}(A_{11}^2\overbar{B_{11}^2}-A_{12}^2\overbar{B_{12}^{2}}+B_{11}^2\overbar{A_{11}^2}-B_{12}^{2}\overbar{A_{12}^2}) \\
=&2\left(\Re(A_{11}^2\overbar{B_{11}^2})-\Re(A_{12}^2\overbar{B_{12}^2})\right)
\end{align*}
A quick check using maple and the definitions of $A_{11}$,  $A_{12}$,  $B_{11}$,  $B_{12}$ shows that $\Re(A_{11}^2\overbar{B_{11}^2})=\Re(A_{12}^2\overbar{B_{12}^2})$.
Thus the coefficient $t_0=0$. 

Now consider the coefficient, $t_1$  
\begin{alignat*}{4}
t_1 = &2z_4^2-2&+&(A_{11}\overbar{A_{12}}+\overbar{A_{11}}A_{12}z_1^2)(\overbar{B_{11}^2}z_4^2-B_{12}^2)\\
&&+&(B_{11}\overbar{B_{12}}+\overbar{B_{11}}B_{12}z_2^2)(\overbar{A_{11}^2}z_4^2-A_{12}^2)\\
=&2z_4^2-2&+&\overbar{A_{11}}A_{12}z_1^2(\overbar{B_{11}^2}z_4^2-B_{12}^2)+\overbar{B_{11}}B_{12}z_2^2(\overbar{A_{11}^2}z_4^2-A_{12}^2)\\
&&+&A_{11}\overbar{A_{12}}(\overbar{B_{11}^2}z_4^2-B_{12}^2)+B_{11}\overbar{B_{12}}(\overbar{A_{11}^2}z_4^2-A_{12}^2) \, . \\
\intertext{Using the Modius transformation it can be rewritten as}
t_1 =&2z_4^2-2&+&\overbar{A_{11}}A_{12}(\overbar{B_{12}^2}z_4^2-B_{11}^2)+\overbar{B_{11}}B_{12}(\overbar{A_{12}^2}z_4^2-A_{11}^2)\\
&&+&A_{11}\overbar{A_{12}}z_1^2(\overbar{B_{11}^2}z_4^2-B_{12}^2)+B_{11}\overbar{B_{12}}(\overbar{A_{11}^2}z_4^2-A_{12}^2) \\
=&(z_4^2-1)(2&+&\overbar{A_{11}}A_{12}\overbar{B_{12}^2}+\overbar{B_{11}}B_{12}\overbar{A_{12}^{2}} + A_{11}\overbar{A_{12}}\overbar{B_{11}^2}+B_{11}\overbar{B_{12}}A_{11}^2) \, .
\end{alignat*}
Again, the definitions of $A$ and $B$ imply that the right hand bracket is zero so that $t_1 = 0$ completing our proof.
\end{proof}

The isolated Hadamard matrix is defined by 
\label{Def:Spectral}
\begin{equation}
\label{eq:spectral}
S^{(0)}_6 = \begin{pmatrix}
		1 & 1 & 1 & 1 &1 & 1\\
		1 & 1 & \omega & \omega & \omega^2 & \omega^2\\
		1 & \omega & 1 & \omega^2 & \omega^2 & \omega\\
		1 & \omega & \omega^2 & 1 & \omega & \omega^2\\
		1 & \omega^2 & \omega^2 & \omega & 1 & \omega \\
		1 & \omega^2 & \omega & \omega^2 & \omega & 1 
		\end{pmatrix} \, ,
\end{equation}
where $ \omega = exp\left(\frac{2\pi i}{3}\right)$. It was noted previously that Theorem \ref{Thm:ExtraConstraints} does not hold for $S^{(0)}_6$ \cite{matolcsi}. 

\section{Mutually unbiased Hadamards and Linear Programming}
The Fourier analytic approach transforms the MUBs existence problem into a \emph{linear} program. The functions $G(\gamma)$ and $F(\gamma)$ defined in Eqs. (\ref{Def:bigG}) and (\ref{Def:bigF}) respectively are treated as variables for each instance of $\gamma$. The constraints on the Hadamard matrices forming a complete set of MUBs, Eqs. (\ref{cond:1}) - (\ref{cond:4}) are now \emph{linear} constraints. 

We aim to prove a structural result on the matrix entries - or show that the program cannot be satisfied. For example, suppose we minimise the variable $G(\rho)$ where $\rho$ is any permutation of $(d,-d,0,\dots,0)$. We know from Eq. (\ref{cond:4}) that $G(\rho)\leq d^3$ so if the linear program returns $G(\rho)\geq d^3$ all elements of the matrices are roots of unity since this is the only way to achieve $G(\rho)=d^3$. Matolcsi et al. use this linear program to prove that the known complete sets of MUBs in dimensions $d=2-5$ are unique.

The functions $g(\gamma)$, $G(\gamma)$ and $F(\gamma)$ take as inputs $\gamma \in \mathbb{Z}^{D}$. We treat each function with a fixed input as a variable in the linear program. In order to limit the size of the search space, we consider a subset, $\Gamma \subset \mathbb{Z}^{D}$, of possible elements $\gamma$. There is considerable flexibility in the choice of $\Gamma$, the aim being to restrict the number of possible $\gamma$ while still including a useful set of equations. We use the set
\begin{defn}
\begin{align}
\Gamma_l &= \left\{\gamma \in \mathbb{Z}^D| |\gamma_1| + \dots + |\gamma_D|\le l \right\} 
\end{align}
for $l\geq 2D$.
\end{defn}

Using the constraints defined in Eqs. (\ref{Def:smallg})-(\ref{Def:bigF}) it is easy to show that for $d=3$, $4$, $5$ that a complete set of MUBs constrains Hadamard matrices whose elements are $d$th roots of unity only  \cite{matolcsi+12}. Unfortunately, running the same program in dimension six does not yield such a structural result.

In order to prove non-trivial structural results for dimension six, we must add more conditions to the linear program. Theorem 1 provides such a condition for the Karlsson family which may allow us to exclude a complete set of MUBs consisting of Karlsson Hadamards alone. Constructing the linear program, we aim to show that there no maximal set of mutually unbiased Karlsson matrices. Such a result would follow if, either we find the constraints contradict themselves so a result is infeasible, or the linear program concludes that the matrices have elements that are 6th roots (or 72nd roots) of unity. A previous brute force computer search has shown that such a set does not exist \cite{bengtsson+07}.

For these proofs to work it is important to have a sufficiently large $l$ in the gamma space, most of the linear programs just required $l=2 d$ as then $\rho \in \Gamma$, however some required a larger $l$. For the non-existence proof you want to have any many conditions as possible as it is difficult to know where a contradiction would occur.

The linear program was written in MATLAB as this has a fast linear program solver. It uses an interior point method for solving large problems such as these, which is a altered form of the Newton method. MATLAB also makes use of sparse matrices for equality and inequality relations, this is vital as the full matrices would have used many tens of GB of RAM while the sparse representation required only tens of MB. 

As the size of the search space $\Gamma_l$ is increased the linear program began to fail due to precision requirements. We were able to resolve this by using the simplex algorithm but unfortunately the time required to solve the linear program increases dramatically (for example several days of computing time were required for problems such with $d=6$ and $l=30$). 
Unfortunately for all sizes of the search space we tried the LP did not return a useful result such as a contradiction or a tightly constrained value of $F(\rho)$. Our results in dimension 6 are summarised in table \ref{table:res}. We show the values of the function $F(\rho)$ for two input values  $\rho_1=(6,-6,0,\dots,0)$ and $\rho_2=(12,-12,0,\dots,0)$ for varying sizes of the search space defined by $l$.
\begin{table}[h]
\begin{center}
    \begin{tabular}{ | l | l | l |}

    \hline
    $l$ & $F(\rho_1)$ & $F(\rho_2)$  \\ \hline
    22 & 71.01 & n/a \\ \hline
    24 &50.02 &  95.62\\ \hline
    26 & 70.52 & 76.10\\ \hline
    28 & 55.81 & 107.72 \\ \hline
    30 & 60.78 & 57.76  \\
    \hline
    
    \end{tabular}
    
\end{center}
\caption{ Table of lower bounds for $F(\rho)$ found by the linear program for varying sizes of the search space $\Gamma_l$ defined in the text.}
 \label{table:res}
\end{table}\\
\noindent

\section{Conclusion}
Proving the non-existence of a complete set of MUBs in any dimension is a highly non-linear. The great strength of the Fourier analytic approach is that it transforms the problem into a linear one. We have continued this approach by proving an additional constraint on a large class of six dimensional Hadamard matrices called the Karlsson family. Although we where not able to derive any further structural results from the linear program we believe that this result together with further ideas is a promising approach to this long standing open problem.

\section*{Acknowledgements}
We would like to thank Mate Matolcsi for helpful comments and sharing his insights. AM was supported by the Heilbronn Institute through a summer research bursary.

\end{document}